\documentclass[11pt]{article}
\usepackage{amsmath,amssymb,amsbsy,amsfonts,amsthm,latexsym,
               amsopn,amstext,amsxtra,euscript,amscd,color}
\def \F {{\mathbb F}}
\def \Q {{\mathbb Q}}
\def \Z {{\mathbb Z}}

\def \C {{\mathbb C}}
\def \V {{\mathbb V}}
\def \o {{\omega}}
\def \Tr {{\rm Tr_n}}
\def \T {{\rm Tr}}

\newtheorem{theorem}{Theorem}[section]

\newtheorem{lemma}[theorem]{Lemma}
\newtheorem{proposition}[theorem]{Proposition}
\newtheorem{remark}[theorem]{Remark}

\newtheorem{corollary}[theorem]{Corollary}

\def\cB{{\mathcal B}}

\def\cH{{\mathcal H}}

\def\cW{{\mathcal W}}

\def\cGB{\mathcal{GB}}

\def\cc{{\bf c}}
\def\uu{{\bf u}}

\def\xx{{\bf x}}

\def\00{{\bf 0}}
\def\11{{\bf 1}}
\def\+{\oplus}

\def \F {{\mathbb F}}
\def \Q {{\mathbb Q}}
\def \Z {{\mathbb Z}}

\def \V {{\mathbb V}}
\def \o {{\omega}}
\def \Tr {{\rm Tr_n}}
\def \T {{\rm Tr}}

\newcommand{\fun}[3]{{{#1}\,:\,{#2}\,\rightarrow\,{#3}}}
\newcommand{\zetak}[1][2^k]{\zeta_{{#1}}}
\newcommand{\GWa}[3][]{{\mathcal H}^{(#2)}_{#3#1}}
\newcommand{\Wa}[2][]{{\mathcal W}_{#2#1}}
\newcommand{\GF}[2][2]{{\mathbb F}_{{#1}^{#2}}}
\newcommand{\VV}[1]{{\mathbb V}_{#1}}
\newcommand{\ZZ}[1]{{\mathbb Z}_{#1}}
\newcommand{\QQ}{{\mathbb Q}}

\newcommand{\rr}{{\bf r}}

\begin{document}


\title{\huge\bf  Decomposing generalized bent and hyperbent functions}

\author{\Large  Thor Martinsen$^1$, Wilfried Meidl$^2$, \and \Large 
Sihem Mesnager$^3$,  Pantelimon St\u anic\u a$^1$
\vspace{0.4cm} \\
\small $^1$Department of Applied Mathematics, \\
\small Naval Postgraduate School, Monterey, CA 93943-5212, U.S.A.;\\
\small Email: {\tt \{tmartins,pstanica\}@nps.edu}\\
\small $^2$Johann Radon Institute for Computational and Applied Mathematics,\\
\small Austrian Academy of Sciences, Altenbergerstrasse 69, 4040-Linz, Austria;\\
\small Email: {\tt meidlwilfried@gmail.com}\\
\small $^3$ Department of Mathematics, \\
\small Universities of Paris VIII and XIII and Telecom ParisTech,\\
\small LAGA, UMR 7539, CNRS, Sorbonne Paris Cit\'e;\\
\small Email: {\tt smesnager@univ-paris8.fr}
}

\date{\today}
\maketitle
\thispagestyle{empty}

\begin{abstract}
In this paper we introduce generalized hyperbent functions from $\F_{2^n}$ to $\Z_{2^k}$, and investigate decompositions of generalized (hyper)bent functions.
We show that generalized (hyper)bent functions from $\F_{2^n}$ to $\Z_{2^k}$ consist of components which are generalized (hyper)bent functions
from $\F_{2^n}$ to $\Z_{2^{k^\prime}}$ for some $k^\prime < k$. For odd $n$, we show that the Boolean functions associated to a generalized bent
function form an affine space of semibent functions. This complements a recent result for even $n$, where the associated Boolean functions are bent.
\end{abstract}


\section{Introduction}

Let $\V_n$ be an $n$-dimensional vector space over $\F_2$ and for an integer $q$, let $\Z_q$ be the ring of integers modulo $q$.
Let $\Re(z)=\alpha$ and $\Im(z)=\beta$ be the real and imaginary parts of a complex number $z=\alpha+\beta i$, respectively.
For a {\it generalized Boolean function} $f:\V_n\to \Z_q$ we define the {\it generalized Walsh-Hadamard transform} to be the complex valued function
\[ \mathcal{H}^{(q)}_f(\uu) = \sum_{\xx\in \V_n}\zeta_q^{f(\xx)}(-1)^{\langle\uu,\xx\rangle}, \]
where $\zeta_q = e^{\frac{2\pi i}{q}}$ and $\langle\uu,\xx\rangle$ denotes a (nondegenerate) inner product on $\V_n$
(we often use $\zeta$, $\cH_f$, instead of $\zeta_q$, respectively, $\cH_f^{(q)}$,  when $q$ is fixed).
For $q=2$, we obtain the usual {\it Walsh-Hadamard transform}
\[ \mathcal{W}_f(\uu) = \sum_{\xx\in \V_n}(-1)^{f(\xx)}(-1)^{\langle\uu,\xx\rangle}. \]
If $\V_n=\F_2^n$, the vector space of the $n$-tuples over $\F_2$, we use the conventional dot product $\uu\cdot\xx$ for $\langle\uu,\xx\rangle$.
The standard inner product of $u,x\in\F_{2^n}$ is $\Tr(ux)$, where $\Tr(z)$ denotes the absolute trace of $z\in\F_{2^n}$.
Most of our general results we will present in the notation of $\V_n = \F_2^n$. For results where we emphasize hyperbent properties we require
$\V_n = \F_{2^n}$.

We use the notations as in \cite{mms0,mms,smgs}. We denote the set of all generalized Boolean functions by $\mathcal{GB}_n^q$ and when $q=2$, by $\mathcal{B}_n$.
A function $f:\V_n\rightarrow\Z_q$ is called {\em generalized bent} ({\em gbent}) if $|\mathcal{H}_f(\uu)| = 2^{n/2}$ for all $\uu\in \V_n$.
We recall that a function $f$ for which $|\mathcal{W}_f(\uu)| = 2^{n/2}$ for all $\uu\in \V_n$ is called a {\em bent} function, which only exist for even $n$ since
$\mathcal{W}_f(\uu)$ is an integer. Further, recall that $f\in\mathcal{B}_n$, $n$ odd, is called {\it semibent} if $|\mathcal{W}_f(\uu)| \in \{0,2^{(n+1)/2}\}$ for all $\uu\in \V_n$. A jubilee survey paper on bent functions giving an historical perspective, and making pertinent
connections to designs, codes and cryptography is \cite{CarletMesnagerDCC2015}. A book devoted especially to bent functions and containing a complete survey 
(including variations, generalizations and applications) is \cite{MesnagerBook}.

In Section \ref{prelim} we recall some results which are of importance
to our considerations and will be used in the following sections.  In
Section \ref{hyper} we introduce generalized hyperbent functions, and
show hyperbentness for classes of gbent functions introduced in
\cite{mms}, which can be seen as generalized Dillon's $PS$ functions.
In Section \ref{decomp} we investigate decompositions of generalized
(hyper)bent functions.  We show that generalized (hyper)bent functions
from $\V_n$ to $\Z_{2^k}$ consist of components which are generalized
(hyper)bent functions from $\V_n$ to $\Z_{2^{k^\prime}}$ for some
$k^\prime < k$. For odd $n$, we show that the Boolean functions
associated to a generalized bent function form an affine space of
semibent functions. This complements a recent result for even $n$,
where the associated Boolean functions are bent.

\section{Preliminaries}
\label{prelim}

We begin by collecting some results which we will subsequently use in the paper. We start with a lemma, which is Proposition 3 in~\cite{mms0}.
\begin{lemma}
\label{valdis}
Let $n=2m$ be even, and for a function $f:\V_n\rightarrow\Z_{2^k}$ and $\uu\in\V_n$, let $f_{\uu}(\xx) = f(\xx)+2^{k-1}(\uu\cdot\xx)$, and let
$b_j^{(\uu)} = |\{\xx\in\V_n\;:\;f_{\uu}(\xx) = j\}|$, $0\le j\le 2^k-1$. Then $f$ is gbent if and only if for all $\uu\in\V_n$
there exists an integer $\rho_{\uu}$, $0\le \rho_{\uu} \le 2^{k-1}-1$, such that
\[ b^{(\uu)}_{2^{k-1}+\rho_{\uu}} = b^{(\uu)}_{\rho_{\uu}}\pm 2^m\; \mbox{and}\; b^{(\uu)}_{2^{k-1}+j} = b^{(\uu)}_j,\,\mbox{for}\; 0\le j \le 2^{k-1}-1, j\ne \rho_{\uu}. \]
\end{lemma}
%
%
In~\cite{mms0} it is shown  that, similar to bent functions (in even and odd characteristic), the value set of $\mathcal{H}_f^{2^k}$ is quite restricted.
\begin{proposition}
If $f\in\mathcal{GB}_n^{2^k}$ is gbent, then
\[ \mathcal{H}_f^{2^k}(\uu) = 2^{n/2}\zeta_{2^k}^{f^*(\uu)} \]
for some function $f^*\in\mathcal{GB}_n^q$, except for $n$ odd and $k=2$, in which case we have
\[ \mathcal{H}_f^4(\uu) = 2^{\frac{n-1}{2}}(\pm 1\pm i). \]
\end{proposition}

In accordance with the terminology for classical bent functions we say that gbent functions are regular (except for the case when $n$ is odd
and $k=2$), and we call the function $f^*$ the {\em dual} of $f$. With the standard proof for bent functions one can show that the dual $f^*$
is also gbent and $(f^*)^* = f$.

Let $f\in \mathcal{GB}_n^{2^k}$, then we can represent $f$ uniquely as
\[ f(\xx) = a_1(\xx) + 2a_2(\xx) + \cdots + 2^{k-1}a_k(\xx) \]
for some Boolean functions $a_i$, $1\le i\le k$. The nature of these Boolean functions when $f$ is gbent has been one of the main topics
in research on gbent functions. In the next proposition and the following remark we summarize some main results on these Boolean functions.
\begin{proposition}
\label{gbebe}
Let $f(\xx)$ be a gbent function in $\mathcal{GB}_n^{2^k}$, $k>1$, (uniquely) given as
\[ f(\xx) = a_1(\xx) + 2a_2(\xx) + \cdots + 2^{k-2}a_{k-1}(\xx) + 2^{k-1}a_k(\xx), \]
$a_i\in\mathcal{B}_n$, $1\le i\le k$, and for $\mathbf{c} = (c_1,c_2,\ldots,c_{k-1}) \in \F_2^{k-1}$, let
$g_{\mathbf{c}}$ be the Boolean function
\begin{equation}
\label{gc}
g_{\mathbf{c}}(\xx) = c_1a_1(\xx) \+ c_2a_2(\xx) \+ \cdots \+ c_{k-1}a_{k-1}(\xx) \+ a_k(\xx).
\end{equation}
\begin{itemize}
\item[$(i)$] \textup{\cite{mms0}} If $n$ is even, then for all $\mathbf{c}\in \F_2^{k-1}$ the Boolean function $g_{\mathbf{c}}$ is a bent function.
\item[$(ii)$] \textup{\cite{mms0, ST09, sgcgm}} If $n$ is odd, and $k=2,3,4$, then all Boolean functions $g_{\mathbf{c}}$, $\mathbf{c}\in \F_2^{k-1}$,
are semibent.
\end{itemize}
\end{proposition}
\begin{remark}
When $n$ is even, then $a_1(\xx) + 2a_2(\xx) \in \mathcal{GB}_n^4$ is gbent if and only if $a_1$ and $a_1\+a_2$ are bent $($see~\textup{\cite{ST09}}$)$.
Sufficient conditions on the gbentness of $f\in\mathcal{GB}_n^{2^k}$ are also known for $k=2$ when $n$ is odd, and in general for $k=3,4$
$($see~\textup{\cite{mms0, ST09, sgcgm}}$)$.
\end{remark}
Another result about the decomposition of gbent functions is the following theorem of ~\cite{mms0}.
\begin{theorem}[\textup{\cite[Theorem 20]{mms0}}]
\label{k,k-1Thm}
Let $f\in\mathcal{GB}_n^{2^k}$ with $f(\xx)=g(\xx)+2h(\xx)$, $g\in\cB_n$, $h\in\cGB_n^{2^{k-1}}$.
If $n$ is even, then the following statements are equivalent.
\begin{itemize}
\item[$(i)$] $f$ is gbent in $\cGB_n^{2^{k}}$;
\item[$(ii)$] $h$ and $h+2^{k-2}g$ are both gbent in $\cGB_n^{2^{k-1}}$ with $\cH_{h+2^{k-2}g}(\uu)=\pm \cH_{h}(\uu)$
for all $\uu\in \V_n$.
\end{itemize}
If $n$ is odd, then $(ii)$ implies $(i)$.
\end{theorem}
%
%
\begin{remark}
\label{r4odd}
In the proof of~\textup{\cite[Theorem 20]{mms0}} it is moreover shown that if $h$ and $h+2^{k-2}$ are gbent, then $f$ is gbent if and only if
$\cH_{h+2^{k-2}g}(\uu)=\pm \cH_{h}(\uu)$ for all $\uu\in \V_n$.
As one of our achievements here, in our Corollary~$\ref{iff-cor}$ we will show that $(i)$ and $(ii)$ in Proposition~$\ref{k,k-1Thm}$ are equivalent also when $n$ is odd.
\end{remark}

\section{Generalized hyperbent functions}
\label{hyper}

Let $f$ be a Boolean function from $\F_{2^n}$ to $\F_2$, and let $1\leq i\leq n$ be an integer with $\gcd(2^n-1,i) = 1$.
The {\it extended Walsh-Hadamard transform} $\cW_{f,i}$ is the integer valued function
\[ \cW_{f,i}(u)=\sum_{x\in \F_{2^n}}(-1)^{f(x)}(-1)^{\Tr(u x^i)}. \]
Recall that $f$ is called {\it hyperbent} if $|\cW_{f,i}(u)|=2^{n/2}$, for all $1\leq i\leq n$ with $\gcd(2^n-1,i) = 1$.
For background on hyperbent functions we refer to the articles \cite{Car06,cg,yg}. 

In this section we introduce the concept of hyperbent functions for generalized Boolean functions, and show the generalized
hyperbentness for a class of gbent functions presented in \cite{mms}. For a function $f\in\mathcal{GB}_n^{2^k}$ and an integer
$1\leq i\leq n$ with $\gcd(2^n-1,i) = 1$, we define the {\it extended generalized Walsh-Hadamard transform} $\cH^{(2^k)}_{f,i}$
as a natural extension of $\cW_{f,i}$ as
\[ \cH_{f,i}^{(2^k)}(u)=\sum_{x\in \F_{2^n}}\zeta_q^{f(x)}(-1)^{\Tr(u x^i)}, \]
and call $f$ a {\it generalized hyperbent} ({\it g-hyperbent}) function if $|\cH^{(2^k)}_{f,i}(u)|=2^{n/2}$, for all $1\leq i\leq n$
with $\gcd(2^n-1,i) = 1$.
%

In~\cite{Car06} Carlet and Gaborit proved that all functions in the class of ${PS}_{ap}$ are hyperbent.
We proceed similarly for a class of gbent functions from $\mathcal{GB}_{2n}^{2^k}$ presented in \cite{mms}, which can be seen as
a function in a generalized ${PS}_{ap}$ class. We recall the functions in the next proposition.
We use the convention that $\frac{y'}{y} = 0$ if $y = 0$.
\begin{proposition}[\textup{\cite[Theorem 1]{mms}}]
\label{exPSk}
Let $g_j:\F_{2^n}\rightarrow\F_2$, $0\leq j\leq k-1$, be
Boolean functions with $g_j(0)=0$ and $\displaystyle\sum_{t\in\F_{2^n}}\zeta^{\sum_{j=0}^{k-1} 2^j g_j(t)}=0$.
Then the function
$f:\F_{2^n}\times\F_{2^n}\rightarrow\Z_{2^k}$ given by
$ \displaystyle f(y',y) = \sum_{j=0}^{k-1} 2^j g_j\left({y'}/{y}\right)$
is a gbent function with the dual
$\displaystyle f^*(y',y) = \sum_{j=0}^{k-1} 2^j g_j\left({y}/{y'}\right)$.
\end{proposition}
To show that these functions are g-hyperbent, we start with some preliminary considerations.
Let $\o$ be any element in $\F_{2^n} \setminus \F_{2^{n/2}}$, then $\F_{2^n} = \F_{2^{n/2}} + \o \F_{2^{n/2}}$. Furthermore,
every $y \in \F_{2^{n/2}}$ satisfies $y^{2^{n/2}} = y$, therefore ${\T(y) = 0}$ for $y\in\F_{2^{n/2}}$.
With the inner product on $\F_{2^n}$ defined by $\langle{y,y'}\rangle = \T(y y')$, the subspace
$\F_{2^{n/2}}$ is orthogonal to itself. Therefore,
\begin{equation}
\label{indi}
\sum_{y\in\F_{2^{n/2}}} (-1)^{\T(\lambda y)} =
\begin{cases}
0 & \text{if }\lambda \notin \F_{2^{n/2}} \\
2^{n/2} & \text{if } \lambda \in \F_{2^{n/2}}
\end{cases}
 = 2^{n/2}\,{\bf 1}_{\F_{2^{n/2}}}(\lambda).
\end{equation}
\begin{theorem}
The function $f$ in Proposition~\textup{\ref{exPSk}} is g-hyperbent.
\end{theorem}
\begin{proof}
We let $g(y'/y):=f(y',y)$. Analogous to Carlet and Gaborit's proof, for an integer $i$ coprime to $2^n-1$,
we write (using $x:=y'+\omega y$, $z:=\frac{y'}{y}$)
\begin{align*}
\cH_{f,i}^{(q)}(u) & = \sum_{x\in\F_{2^n}}\zeta^{f(x)}(-1)^{Tr(ax^i)}\\
&=\sum_{y,y'\in\F_{2^{n/2}}}\zeta^{g\left(\frac{y'}{y}\right)}(-1)^{\T(a(y'+\o y)^{i})}\\
&=\sum_{y\in\F_{2^{n/2}}^{*},y'\in\F_{2^{n/2}}}\zeta^{g\left(\frac{y'}{y}\right)}(-1)^{\T(ay^{i}(z+\o)^{i})}
+ \sum _{y'\in\F_{2^{n/2}}}\zeta^{g(0)\oplus\T(ay'^i)}.
\end{align*}
With $(\ref{indi})$ we obtain
\begin{align*}
\cH_{f,i}^{(q)}(u)
=&\sum_{z\in\F_{2^{n/2}}}\zeta^{g(z)}
\sum_{y\in\F_{2^{n/2}}^{*}}(-1)^{\T(a(z +\o)^{i}y^i)}+\zeta^{g(0)}2^{n/2}\cdot{\bf 1}_{\F_{2^{n/2}}}(a) \\
=&\sum_{z\in\F_{2^{n/2}}}\zeta^{g(z)}\sum_{y\in\F_{2^n/2}}(-1)^{\T(a(z+\o)^{i}y^i)}\\
&-\sum_{z\in\F_{2^{n/2}}}\zeta^{g(z)} + \zeta^{g(0)}2^{n/2}\cdot{\bf 1}_{\F_{2^{n/2}}}(a).
\end{align*}

Substituting $\displaystyle g(z) =\sum_{j=0}^{k-1}2^{j}g_j(z)$ we have:
\begin{align*}
&\sum_{z\in\F_{2^{n/2}}}\zeta^{\sum_{j=0}^{k-1}2^{j}g_j(z)}\sum_{y\in\F_{2^{n/2}}}(-1)^{\T(a(z+w)^{i}y^i)}\\
&-\sum_{z\in \F_{2^{n/2}}}\zeta^{\sum_{i=0}^{k-1}2^{j}g_j(z)} + \zeta^{g(0)}2^{n/2}\cdot{\bf 1}_{\F_{2^{n/2}}}(a).
\end{align*}

By \cite[Lemma 1]{Car06}, if $a\not\in\F_{2^n}$, then there exists a unique $z$ such that $a(z + \o)^i \in \F_{2^{n/2}}$,
which in turn means that $\T(a(z +\o)^{i}y^i) = 0$,
since $y^i \in \F_{2^{n/2}}$. Hence,
the first term $\sum_z\zeta^{\sum_{j=0}^{k-1} 2^{j}g_j(z)} \sum_y (-1)^{\T(a(z+w)^i y^i)}$  in the above expression equals $\zeta^{\rho}2^{n/2}$ (for some positive integer $\rho$), if $a\notin \F_{2^{n/2}}$
and zero otherwise. Moreover, the second term $\sum_z \zeta^{\sum_{j=0}^{k-1} 2^{j}g_j(z)}$ equals zero by definition, and as previously stated, the last term
$\zeta^{\sum_j 2^{j}g_j(0)\oplus {\T(ay'^i)}}$ equals $\zeta^{g(0)}2^{n/2}$, if $a \in \F_{2^{n/2}}$ and zero otherwise. Therefore, we see that the
entire previously displayed expression equals $\zeta^\rho 2^{n/2}$, for some integer $\rho$, regardless of
whether $a \in \F_{2^{n/2}}$ or $a \notin \F_{2^{n/2}}$ and therefore, $f$ is g-hyperbent.
\end{proof}

More generally, one can generalize a classical construction of Boolean
hyperbent functions as follows.  We have the multiplicative
decomposition $\GF{n}^\star=\GF{m}^\star \times U$ where $U$ is a
cyclic subgroup of $\GF{n}^\star$ of order $2^m+1$, $m=\frac n2$.  Let
$\fun f{\VV n}{\ZZ{2^k}}$ be such that $f$ is constant on each coset
$a\GF m^\star$ for any $a\in U$.  Then

\begin{theorem}
  Let $k\geq 3$. Then, $f$ is g-hyperbent if and only if
  $\sum_{u\in U}\zeta_{2^k}^{f(u)}=\zeta_{2^k}^{f(0)}$.
\end{theorem}
\begin{proof}
  \begin{align*}
    \GWa{2^k}{f}(a) 
    &= \sum_{x\in\GF n}\zetak^{f(x)}(-1)^{\Tr(ax^i)}\\
    &= \zetak^{f(0)} + \sum_{u\in U} \zetak^{f(u)}\sum_{y\in\GF{m}^\star}(-1)^{\Tr(au^iy^i)}\\
    &= \zetak^{f(0)} - \sum_{u\in U} \zetak^{f(u)} + \sum_{u\in U} \zetak^{f(u)}\sum_{y\in\GF{m}}(-1)^{\T_m\big(\T_m^n(au^i)y^i\big)}.
  \end{align*}
  Now,
  \begin{align*}
    \sum_{y\in\GF{m}}(-1)^{\T_m\big(\T_m^n(au^i)y^i\big)} = \sum_{y\in\GF{m}}(-1)^{\T_m\big(\T_m^n(au^i)y\big)},
  \end{align*}
  since $\gcd(i,2^m-1)=\gcd(i,2^n-1)=1$. Observe that the equation
  $\T_m^n(au^i)=au^i+a^{2^m}u^{-i}=0$ has a unique solution $u_a$
  in $U$ for every $a\not=0$. Thus, if $a\not=0$,
  \begin{align*}
    \GWa{2^k}{f}(a) 
    &= \zetak^{f(0)} - \sum_{u\in U} \zetak^{f(u)} + 2^m\zetak^{f(u_a)}.
  \end{align*}
  On the other hand,
  \begin{align}
  \label{eq:!=0}
    \GWa{2^k}{f}(0) 
    &= \zetak^{f(0)} - \sum_{u\in U} \zetak^{f(u)} + 2^m\sum_{u\in U} \zetak^{f(u)}.
  \end{align}
  Suppose that $\sum_{u\in U}\zeta_{2^k}^{f(u)}=\zeta_{2^k}^{f(0)}$. Then
  \begin{align}
  \label{eq:=0}
    \GWa{2^k}{f}(a) = 2^m\zetak^{f(u_a)}\quad\text{and}\quad\GWa{2^k}{f}(0) = 2^m\zetak^{f(0)}.
  \end{align}
  Conversely, suppose that $f$ is g-hyperbent. Then, for $a\not=0$,
  \begin{align*}
    \zetak^{f(0)} - \sum_{u\in U} \zetak^{f(u)} + 2^m\zetak^{f(u_a)} = 2^m \zetak^{\rho}
  \end{align*}
  and
  \begin{align*}
    \zetak^{f(0)} - \sum_{u\in U} \zetak^{f(u)} + 2^m\sum_{u\in U} \zetak^{f(u)}= 2^m\zetak^{\phi},
  \end{align*}
  for some $\rho\in\ZZ {2^k}$ and $\phi\in\ZZ{2^k}$. Set
  $N_r^+=\vert\{u\in U\mid f(u)=r\}\vert$,
  $N_r^-=\vert\{u\in U\mid f(u)=r+2^{k-1}\}\vert$ and
  $N_r=N_r^+-N_r^-$ for $r\in\ZZ {2^{k-1}}$ and, for $e\in\ZZ{2^k}$,
  $e=\rr(e)+2^{k-1}s(e)$. Then, equation~(\ref{eq:!=0}) can be
  rewritten as
  \begin{align*}
    &\sum_{r\in\ZZ {2^{k-1}}\setminus\{\rr(\rho),\rr(f(0)),\rr(f(u_a))\}} N_r\zetak^r + \left(-N_{\rr(f(u_a))} +2^m(-1)^{s(f(u_a))}\right)\zetak^{\rr(f(u_a))}\\
    &\qquad+\left(-N_{\rr(f(0))}+(-1)^{s(f(0))}\right)\zetak^{\rr(f(0))}
      + \left(-N_{\rr(\rho)}-2^m(-1)^{s(\rho)}\right)\zetak^{\rr(\rho)}=0.
  \end{align*}
  Thus, since $\{\zetak^{\rho},\,0\leq \rho\leq2^{k-1}-1\}$ is a basis
  of $\QQ(\zetak)$,
  \begin{align*}
    &N_\rr = -N_{\rr(f(u_a))} +2^m(-1)^{s(f(u_a))} = -N_{\rr(f(0))}+(-1)^{s(f(0))} \\
    &\quad= -N_{\rr(\rho)}-2^m(-1)^{s(\rho)}=0,
  \end{align*}
  for every
  $r\in\ZZ {2^{k-1}}\setminus\{\rr(\rho),\rr(f(0)),\rr(f(u_a))\}$.
  Therefore
  \begin{align*}
    \sum_{u\in U}\zetak^{f(u)}&=\sum_{r\in\ZZ {2^{k-1}}}N_r\zetak^r \\ &=2^m(-1)^{s(f(u_a))}\zetak^{\rr(f(u_a))}+(-1)^{s(f(0))}\zetak^{\rr(f(0))}
    -2^m(-1)^{s(\rho)}\zetak^{\rr(\rho)}.
  \end{align*}
  Thus
  \begin{align*}
    &2^m(-1)^{s(f(0)}\zetak^{\rr(f(0))}-2^m(-1)^{s(\phi)}\zetak^{\rr(\phi)}\\
    &\quad + (2^n-2^m) \left((-1)^{s(fu_a)}\zetak^{f(u_a)}-(-1)^{s(\rho)}\zetak^{\rho}\right)=0.
  \end{align*}
  Therefore, $\zetak^{f(0)}=\zetak^{\phi}$ and
  $\zetak^{f(u_a)}=\zetak^{\rho}$ proving that
  $\sum_{u\in U}\zetak^{f(u)}=\zetak^{f(0)}$.
\end{proof}

\section{Decomposition of gbent and g-hyperbent functions}
\label{decomp}

Let $f\in \mathcal{GB}_n^{2^k}$ be a gbent function. In this section we continue analyzing the nature of Boolean and generalized
Boolean functions in $\mathcal{GB}_n^{2^{k^\prime}}$, $k^\prime < k$, of which the gbent function $f$ is (in some sense) composed.

Firstly, any function $\fun {f}{\GF n}{\ZZ {2^k}}$ can be uniquely
decomposed as
\begin{align*}
  f(x) = \sum_{j=0}^{k-1}2^j f_j,
\end{align*}
where the $f_j$'s are Boolean functions. It has been recalled in
Proposition \ref{gbebe} that, when $n$ is even, if $f$ is gbent then
all its ``components'' $f_j$ are bent functions. In fact, one can
extend the previous results to g-hyperbent functions. To this end, we
make some preliminary remarks that will help us in our analysis. Recall that when $\gcd(i,2^n-1)=1$, then the
extended Walsh transform of $f$ is
\begin{align*}
  \GWa[,i]{2^k}{f}(a) &= \sum_{x\in\GF n}\zetak^{f(x)}(-1)^{\Tr(ax^i)}\\
&= \sum_{x\in\GF n}\zetak^{f(x^j)}(-1)^{\Tr(ax)} = \GWa{2^k}{f(x^j)}(a)
\end{align*}
where $j$ is the inverse of $i$ in $\ZZ{2^n-1}$. Now, saying that $f$
is g-hyperbent is equivalent to say that $f(\xx^j)$ is g-bent for
every $j$ coprime with $2^n-1$. Thus, for $k\geq 3$,
\begin{align}
  \label{eq:regular}
  \GWa[,i]{2^k}{f}(a) = 2^{\frac n2}\zetak^\rho
\end{align}
for some $\rho\in\ZZ{2^k}$. Now, Observe that
\begin{align}\label{eq:extension}
  \zetak^{f(x)} 
  = \prod_{j=0}^{k-1} \zetak^{2^jf_j(x)}
  = \prod_{j=0}^{k-1} \left(\frac{1+\zetak^{2^j}}2+\frac{1-\zetak^{2^j}}2(-1)^{f_j(x)}\right).
\end{align}
Set
\begin{align*}
  Q(X_1,\dots,X_{k-1}) 
  &= \prod_{j=0}^{k-1} \left(\frac{1+\zetak^{2^j}}2+\frac{1-\zetak^{2^j}}2X_j\right)\\
  &= 2^{-k}\prod_{j=0}^{k-1} \sum_{c\in\F_2}\left(\left(1+\zetak^{2^j+c2^{k-1}}\right)X_j^c\right)\\
  &= 2^{-k}\sum_{c\in\F_2^{k}} \left(\prod_{j=0}^{k-1}\left(1+\zetak^{2^j+c_j2^{k-1}}\right)\right)\prod_{j=0}^{k-1}X_j^{c_j}.
\end{align*}
Set
$A_c=2^{-k}\prod_{j=0}^{k-1}\left(1+\zetak^{2^j+c_j2^{k-1}}\right)$. Then
\begin{align}
  \label{eq:decompositionzetak}
  \zetak^{f(x)} 
   = Q\left((-1)^{f_0(x)},\dots,(-1)^{f_{k-1}(x)}\right)
   = \sum_{c\in\F_ 2^{k}} A_c (-1)^{\sum_{j=0}^{k-1}c_jf_j(x)}.
\end{align}
Then, we have the folowing theorem.
\begin{theorem}
\label{ghyperbent}
Let $f:\V_n\rightarrow\Z_{2^k}$, $n$ even. Then $f$ is a g-hyperbent
function given as
$f(\xx) = a_1(\xx)+2a_2(\xx)+\cdots+ 2^{k-1}a_{k}(\xx)$ if and only
if, for each $\cc\in\F_2^{k-1}$, the Boolean function $f_\cc$ defined
as
\[ f_\cc(\xx) =
c_1a_1(\xx)\+c_2a_2(\xx)\+\cdots\+c_{k-1}a_{k-1}(\xx)\+a_{k}(\xx) \]
is a hyperbent function.
\end{theorem}
\begin{proof}
  Let $i$ be coprime with $2^n-1$.  According
  to~(\ref{eq:decompositionzetak}),
  \begin{align*}
    \GWa[,i]{2^k}{f} (a) 
    &= \sum_{x\in\GF n}  \sum_{c\in\F_ 2^{k}} A_c (-1)^{\sum_{j=0}^{2^k-1}c_jf_j(x) + \Tr(ax^i)}\\
    &=  \sum_{c\in\F_ 2^{k}} A_c \Wa[,i]{f_\cc}(a).
  \end{align*}
  Now,
\begin{align*}
  2^k A_c 
  &= \prod_{j=0}^{k-1}\sum_{d\in\ZZ 2}\zetak^{d2^j+dc_j2^{k-1}}
  = \sum_{d\in\F_ 2^{k}}\zetak^{\sum_{j=0}^{k-1}d_j2^j+d_jc_j2^{k-1}}\\
  &= \sum_{d\in\F_ 2^{k}}\zetak^{\sum_{j=0}^{k-2}d_j2^j} (-1)^{\sum_{j=0}^{k-1}d_jc_j}.
\end{align*}
Then
\begin{align*}
  2^kA_c
  &=\sum_{d\in\F_ 2^{k}}\zetak^{\sum_{j=0}^{k-2}d_j2^j+2^{k-1}d_{k-1}} (-1)^{\sum_{j=0}^{k-1}d_jc_j}\\
  &=\left(\sum_{d_{{k-1}}\in\F_ 2} (-1)^{d_{k-1}+c_{k-1}d_{k-1}}\right)\sum_{(d_0,\dots,d_{{k}-2)})\in\F_{2}^{{k-1}}} \zetak^{\sum_{j=0}^{k-2}d_j2^j} (-1)^{\sum_{j=0}^{k-2}d_jc_j}\\
  &=\begin{cases}
    0 & \text{if $c_{k-1}=0$,}\\
    2\sum_{(d_0,\dots,d_{{k}-2)})\in\F_{2}^{k-1}} \zetak^{\sum_{j=0}^{k-2}d_j2^j} (-1)^{\sum_{j=0}^{k-2}d_jc_j} & \text{if $c_{k-1}=1$}.
  \end{cases}
\end{align*}
Define a ``dot product'' over $\F_{2}^{k-1}$ by setting
$c\cdot d=\sum_{j=0}^{l-2}c_jd_j$ for
$c=(c_0,c_1,\dots,c_{k-2})\in\F_2^{k-1}$ and
$d=(d_0,d_1,\dots,d_{k-2})\in\F_2^{k-1}$.  Define the ``canonical
injection'' $\fun{\iota} {\F_2^{k-1}}{\ZZ{2^{k-1}}}$  by
$\iota(c) = \sum_{j=0}^{k-2}c_k2^j$ where $c=(c_0,c_1,\dots,c_{k-2})$.
Then
\begin{align}\label{eq:relationGWaGray}
  \GWa[,i]{2^k}{f} (a) 
  &= \frac{1}{2^{k-1}} \sum_{(c,d)\in\F_ {2}^{k-1}\times\F_{2}^{k-1}} (-1)^{c\cdot d}\zetak^{\iota(d)}\,
    \Wa[,i]{f_\cc}(a).
\end{align}

Suppose now that $\fun f{\GF n}{\ZZ{2^k}}$ is g-hyperbent, so 
 for every
$i$ coprime with $2^n-1$, we have
\begin{align*}
  \GWa[,i]{2^k}{f}(a) = 2^{\frac n2}\zetak^{f^\star_i(a)}
\end{align*}
for some $\fun {f^\star_i}{\GF n}{\ZZ{2^k}}$.  Fix $i$ coprime with
$2^n-1$ and decompose $f^\star_i$ as $f^\star_i = g+2^{k-1}s$ with
$\fun {g}{\GF n}{\ZZ{2^{k-1}}}$ and $\fun {s}{\GF n}{\F_{2}}$ so that
\begin{align*}
  \GWa[,i]{2^k}{f}(a) = 2^{\frac n2}(-1)^{s(a)}\zetak^{g(a)}.
\end{align*}
Then, 
\begin{align}\label{eq:vanishing1}
  \sum_{d\in\F_{2}^{k-1}}\left( \frac{1}{2^{k-1}} \sum_{c\in\F_ {2}^{k-1}} (-1)^{c\cdot d}
  \Wa[,i]{f_c}(a)\right)\zetak^{\iota(d)} - 2^{\frac n2}(-1)^{s(a)}\zetak^{g(a)} = 0.
\end{align}
Now, $\{1,\zetak,\dots,\zetak^{2^{k-1}-1}\}$ being a basis of
$\QQ(\zetak)$,
\begin{align}\label{eq:valeur1}
  \frac{1}{2^{k-1}} \sum_{c\in\F_2^{k-1}} (-1)^{c\cdot d}
  \Wa[,i]{f_\cc}(a) &= \begin{cases}
    0 & \mbox{if $d\not =g(a)$}\\
    2^{\frac n2}(-1)^{s(a)} & \mbox{if $d=g(a)$}.
    \end{cases}
\end{align}
Now, let us invert (\ref{eq:valeur1}). We have for any $\cc\in\F_2^{k-1}$
\begin{align*}
  \Wa[,i]{f_\cc}(a) 
  &= \frac{1}{2^{k-1}}  \sum_{(c,d)\in\F_ 2^{k-1}} (-1)^{(c + \cc)\cdot d}
    \Wa[,i]{f_{c}}(a)\\
  &=\sum_{d\in\F_ 2^{k-1}}(-1)^{\cc\cdot d}\left(\frac{1}{2^{k-1}} 
    \sum_{\underbar c\in\F_ 2^{k-1}} (-1)^{c\cdot d}
    \Wa[,i]{f_c}(a)\right)\\
  &= (-1)^{\cc\cdot g(a)+s(a)}2^{\frac n2},
\end{align*}
for every $a\in\GF n$. Since $i$ is arbitrary in the preceding
calculation, that shows that $f_\cc$ is hyperbent.

Conversely, suppose that, for every $\gcd(i,2^n-1)=1$, there exists
$\fun {g_i}{\GF n}{\ZZ{2^{k-1}}}$ and $\fun{s_i}{\GF n}{\F_{2}}$
such that, for every $c\in\F_{2}^{k-1}$,
\begin{align*}
  \Wa[,i]{f_c}(a) &= 2^{\frac n2}(-1)^{c\cdot \iota^{-1}(g_i(a))+s_i(a)}.
\end{align*}
Thus, for every $\gcd(i,2^n-1)=1$, we have
\begin{align*}
  \GWa[,i]{2^k}{f} (a) 
  &= \frac{1}{2^{k-1}} \sum_{(c,d)\in\F_ 2^{k-1}\times\F_ 2^{k-1}} (-1)^{c\cdot d}\zetak^{\iota(d)}
    \Wa[,i]{f_c}(a)\\
  &= 2^{\frac n2}\cdot\frac{1}{2^{k-1}}\cdot
    \sum_{(c,d)\in\F_2^{k-1}\times\F_ 2^{k-1}} (-1)^{c\cdot d+c\cdot \iota^{-1}(g_i(a))+s_i(a)}\zetak^{\iota(d)}\\
  &= 2^{\frac n2} (-1)^{s_i(a)}\cdot
    \sum_{d\in\F_ 2^{k-1}}\left(\frac{1}{2^{k-1}}\sum_{c\in\F_{2}^{k-1}}(-1)^{c\cdot (d+ \iota^{-1}(g_i(a)))}\right)\zetak^{\iota(d)}\\
  &=  2^{\frac n2} (-1)^{s_i(a)} \zetak^{g_i(a))}
\end{align*}
proving that $f$ is g-hyperbent.
\end{proof}

\begin{remark}
  In the proof of Theorem~\textup{\ref{ghyperbent}}, we have only used the fact
  that the Walsh transform of $f(\xx^i)$ divided by its magnitude is a
  root of unity. The proof of Theorem~\textup{\ref{ghyperbent}} proposes
  therefore an alternate proof of (i) of Proposition~\textup{\ref{gbebe}}. It
  also shows that the g-bentness of $f$ is equivalent to the bentness
  of all the ``component functions'' $f_\cc$.
\end{remark}

We now turn our attention to the case where $n$ is odd and prove the following.

\begin{theorem}
\label{sbent}
Let $f:\V_n\rightarrow\Z_{2^k}$, $n$ odd, be a gbent function given as
$f(\xx) = a_1(\xx)+2a_2(\xx)+\cdots+ 2^{k-1}a_{k}(\xx)$.  If $f$ is
gbent then, for each $\cc\in\F_2^{k-1}$, the Boolean function $g_\cc$
defined as
\[ g_\cc(\xx) = c_1a_1(\xx)\+c_2a_2(\xx)\+\cdots\+c_{k-1}a_{k-1}(\xx)\+a_{k}(\xx) \]
is a semibent function.
\end{theorem}

\begin{proof}
  We know that $2^{-\frac n2}\GWa{2^k}{f}(a)$ is a root of
  unity. Therefore, for every $a\in\GF n$,
  \begin{align*}
    \GWa{2^k}{f}(a) 
     &= 2^{\frac n2}\zetak^{f^\star(a)}
       = 2^{\frac{n-1}2}\, \sqrt{2} \,\zetak^{f^\star(a)},
  \end{align*}
  for some map $\fun {f^\star}{\GF n}{\ZZ{2^k}}$.  Recall now that
  $\mathbb Q(\sqrt 2)\subset\mathbb Q (\zetak)$. Indeed,
  $\sqrt 2=\zetak[8] +\bar{\zetak[8]}=\zetak[8]
  +\zetak[8]^{-1}=\zetak[8]
  +\zetak[8]^{7}=\zetak[8]-\zetak[8]^3=\zetak^{2^{k-3}}-\zetak^{3\cdot
    2^{k-3}}$. Thus
\begin{align*}
  \GWa{2^k}{f}(a) 
  &=  2^{\frac{n-1}2}\left( \zetak^{f^\star(a)+2^{k-3}} - \zetak^{f^\star(a)+3\cdot 2^{k-3}}\right).
\end{align*}
Write $f^\star(a)+2^{k-3}=g_1(a)+2^{k-1}s_2(a) + 2^kt_1(a)$ and
$f^\star(a)+3\cdot 2^{k-3}=g_2(a)+2^{k-1}s_2(a)+2^k t_2(a)$ so that
\begin{align*}
  \GWa{2^k}{f}(a) = 2^{\frac{n-1}2}(-1)^{s_1(a)}\zetak^{g_1(a)} -2^{\frac{n-1}2}(-1)^{s_2(a)} \zetak^{g_2(a)}.
\end{align*}
In the proof of Theorem \ref{ghyperbent}, we have established the
following relation between the Walsh-Hadamard transform of $f$ and the
Walsh transform of its ``component'' $f_\cc$ (take $i=1$ in
(\ref{eq:relationGWaGray}) and recall that $\iota$ is the
``canonical'' injection from $\F_2^{k-1}$ to $\ZZ{2^{k-1}}$ which
sends $(c_0,\dots,c_{k-2})$ to $\sum_{j=0}^{k-2}c_j2^j$), namely,
\begin{align}\label{eq:relationGWaGray}
  \GWa{2^k}{f} (a) 
  &= \frac{1}{2^{k-1}} \sum_{(c,d)\in\F_ {2}^{k-1}\times\F_{2}^{k-1}} (-1)^{c\cdot d}\zetak^{\iota(d)}.
    \Wa{f_\cc}(a)\\
  &= \sum_{d\in\F_2^{k-1}} \left( \frac{1}{2^{k-1}} \sum_{c\in\F_ 2^{k-1}} (-1)^{c\cdot d}
    \Wa{f_c}(a) \right)\zetak^{\iota(d)}.
\end{align}
Then, one has
\begin{align}
  \label{eq:valeur2}
  \frac{1}{2^{k-1}} \sum_{c\in\F_ 2^{k-1}} (-1)^{c\cdot d}
  \Wa{f_c}(a) 
  &= \begin{cases}
    0 & \mbox{if $d\not\in\{g_1(a),g_2(a)\}$}\\
    2^{\frac {n-1}2}(-1)^{s_1(a)} & \mbox{if $d=g_1(a)$}\\
    -2^{\frac {n-1}2}(-1)^{s_2(a)} & \mbox{if $d=g_2(a)$}.
    \end{cases}
\end{align}
Thus
\begin{align*}
  \Wa{f_\cc}(a) 
  &= \frac{1}{2^{k-1}}  \sum_{(c,d)\in\F_ 2^{k-1}\times\F_ 2^{k-1}} (-1)^{(c + \cc)\cdot d}
    \Wa{f_c}(a)\\
  &=\sum_{d\in\F_ 2^{k-1}}(-1)^{\cc\cdot d}\frac{1}{2^{k-1}} 
    \sum_{c\in\F_ 2^{k-1}} (-1)^{c\cdot d}
    \Wa {f_c}(a)\\
  &= \frac{(-1)^{\cc\cdot g_1(a)+s_1(a)}-(-1)^{\cc\cdot g_2(a)+s_2(a)}}{2}\,2^{\frac {n+1}2}
\end{align*}
proving that $f_\cc$ is semibent since
\begin{align*}
  \frac{(-1)^{\cc\cdot g_1(a)+s_1(a)}-(-1)^{\cc\cdot g_2(a)+s_2(a)}}{2}\in\{-1,0,1\}
\end{align*}
for every $a\in\GF n$.
\end{proof}

In the following proposition we decompose a gbent function in
$\mathcal{GB}_n^{2^k}$ into two gbent functions in
$\mathcal{GB}_n^{2^{k^\prime}}$ for some $k^\prime$ smaller than
$k$. We will show the decomposition more general for g-hyperbent
functions, where we consider functions from $\F_{2^n}$ to $\Z_{2^k}$.
The crucial lemma for analyzing the decomposition of $f$ when $n$ is
even, is Lemma~\ref{valdis}.  For instance the proof of
Proposition~\ref{gbebe} $(i)$ is based on this lemma.

We intend to show our results on decompositions of gbent functions for $n$ even and for $n$ odd simultaneously.
Therefore we first deduce a more complex analog of Lemma~\ref{valdis} which is applicable to gbent functions in an odd number of variables.



For $k \ge 3$, let again $\zeta_{2^k}$ be a primitive $2^k$-root of unity. Then $\zeta_{2^k}^{2^{k-3}}$ is a primitive $2^3$-root of unity,
and without loss of generality, we assume that $\zeta_{2^k}^{2^{k-3}} = \zeta_{2^3} = (1+i)/\sqrt{2}$. Recall that for $k\ge 3$ every gbent function is regular, i.e.
for an integer $0\le \rho_\uu\le 2^k-1$ (depending on $\uu$) we have
\begin{align*}
 \mathcal{H}_f^{(2^k)}(\uu) &= 2^{n/2}\zeta_{2^k}^{\rho_\uu} = 2^{n/2} \zeta_{2^k}^{2^{k-3}} \zeta_{2^k}^{\rho_\uu-2^{k-3}}\\
  &= 2^\frac{n-1}{2}(1+i)\zeta_{2^k}^{\rho_\uu-2^{k-3}} \\
 &= 2^\frac{n-1}{2}\zeta_{2^k}^{\rho_\uu-2^{k-3}} + 2^\frac{n-1}{2}\zeta_{2^k}^{2^{k-2}}\zeta_{2^k}^{\rho_\uu-2^{k-3}} \\
 &= 2^\frac{n-1}{2}\zeta_{2^k}^{\rho_\uu-2^{k-3}} + 2^\frac{n-1}{2}\zeta_{2^k}^{\rho_\uu+2^{k-3}}.
\end{align*}
\begin{proposition}
\label{GB-con}
For an odd integer $n$ and $k\ge 3$, let $f$ be a function from $\V_n$ to $\Z_{2^k}$, for $\uu\in\V_n$ let $f_\uu(\xx) = f(\xx) + 2^{k-1}(\uu\cdot\xx)$,
and let $B_\uu(\rho) = \{\xx\in\V_n\,:\,f_\uu(\xx) = \rho\}$. Then $f$ is gbent if and only if for all $\uu\in\V_n$ there exists an integer $\rho_\uu$, $0\le \rho_\uu\le 2^k-1$,
such that
\[
|B_{\uu}(\rho_\uu-2^{k-3}+2^{k-1})| = |B_\uu(\rho_\uu-2^{k-3})|\pm 2^{\frac{n-1}{2}}
\]
and
\[
|B_\uu(\rho_\uu+2^{k-3}+2^{k-1})| = |B_\uu(\rho_\uu+2^{k-3})|\pm 2^{\frac{n-1}{2}}
\]
where in both equations we have the same sign (and the argument of $B_\uu$ is reduced modulo $2^k$), and
\[
|B_\uu(\rho+2^{k-1})| = |B_\uu(\rho)|,
 \]
if $\rho \ne \rho_\uu\pm 2^{k-3}, \rho_\uu\pm 2^{k-3}+2^{k-1}$.
\end{proposition}
\begin{proof}
Let $f$ be a function from $\V_n$ to $\Z_{2^k}$ for which the conditions in the proposition hold. For $\uu\in\V_n$, the generalized Walsh-Hadamard transform at $\uu$ is then
\begin{align*}
\mathcal{H}_f^{(2^k)}(\uu) &= \sum_{\xx\in\V_n}\zeta_{2^k}^{f_\uu(\xx)} = \sum_{\rho=0}^{2^k-1}|B_\uu(\rho)|\zeta_{2^k}^\rho \\
& = (|B_\uu(\rho_\uu-2^{k-3})| - (|B_\uu(\rho_\uu-2^{k-3})|\pm 2^{\frac{n-1}{2}}))\zeta_{2^k}^{\rho_\uu-2^{k-3}} \\ & +
(|B_\uu(\rho_\uu+2^{k-3})| - (|B_\uu(\rho_\uu+2^{k-3})|\pm 2^{\frac{n-1}{2}}))\zeta_{2^k}^{\rho_\uu+2^{k-3}} \\
& = \pm 2^{\frac{n-1}{2}}\zeta_{2^k}^{\rho_\uu-2^{k-3}} \pm 2^{\frac{n-1}{2}}\zeta_{2^k}^{\rho_\uu+2^{k-3}} = 2^{\frac{n-1}{2}}\zeta_{2^k}^{\rho_\uu}\zeta_{2^k}^{2^{k-3}}(\pm i\pm 1) \\
& = 2^{\frac{n-1}{2}}\zeta_{2^k}^{\rho_\uu}\frac{1+i}{\sqrt{2}}(\pm i\pm 1) = 2^{\frac{n-1}{2}}\zeta_{2^k}^{\rho_\uu}\frac{1+i}{\sqrt{2}}\alpha.
\end{align*}
(Here the arguments of $B_\uu$ are reduced modulo $2^k$.)
With $\frac{1+i}{\sqrt{2}}(1+i) = \sqrt{2}i = \sqrt{2}\zeta_{2^k}^{2^{k-2}}$, we get $\mathcal{H}_f^{(2^k)}(\uu) = 2^{n/2}\zeta_{2^k}^{\rho_\uu+2^{k-2}}$ when
$\alpha = 1+i$. Similarly, when $\alpha = -1-i$, $\alpha = 1-i$, respectively $\alpha = -1+i$, for $\mathcal{H}_f^{(2^k)}(\uu)$ we obtain
$2^{n/2}\zeta_{2^k}^{\rho_\uu+2^{k-2}+2^{k-1}}$, $2^{n/2}\zeta_{2^k}^{\rho_\uu}$, respectively $2^{n/2}\zeta_{2^k}^{\rho_\uu+2^{k-1}}$. Therefore $f$ is gbent.

Conversely suppose that $f$ is gbent. As observed above, for $\uu\in\V_n$ we then have
\begin{equation}
\label{Hf1}
\mathcal{H}_f^{(2^k)}(\uu) = 2^\frac{n-1}{2}\zeta_{2^k}^{\rho_\uu-2^{k-3}} + 2^\frac{n-1}{2}\zeta_{2^k}^{\rho_\uu+2^{k-3}},
\end{equation}
for some $0\le \rho_\uu \le 2^k-1$ depending on $\uu$. By the definition of $B_\uu(\rho)$ we moreover have
\begin{align}
\label{Hf2}
\nonumber
\mathcal{H}_f^{(2^k)}(\uu) & = |B_\uu(0)| + |B_\uu(1)|\zeta_{2^k} + \cdots + |B_\uu(2^{k-1}-1)|\zeta_{2^k}^{2^{k-1}-1} \\ \nonumber
& + |B_\uu(2^{k-1})|(-1) + |B_\uu(2^{k-1}+1)|\zeta_{2^k}^{2^{k-1}+1} + \cdots + |B_\uu(2^k-1)|\zeta_{2^k}^{2^k-1} \\ \nonumber
& = (|B_\uu(0)|-|B_\uu(2^{k-1})|) + (|B_\uu(1)|-|B_\uu(2^{k-1}+1)|)\zeta_{2^k} + \cdots \\ & + (|B_\uu(2^{k-1}-1)|-|B_\uu(2^k-1)|)\zeta_{2^k}^{2^{k-1}-1}.
\end{align}
Since $\{1,\zeta_{2^k},\ldots,\zeta_{2^k}^{2^{k-1}-1}\}$ is a basis of $\Q(\zeta_{2^k})$, the conditions in the proposition follow from equations~\eqref{Hf1} and~\eqref{Hf2}.
\end{proof}

%
Let us now explain how to deduce from Proposition \ref{GB-con} a first
result. We include the hyperbent condition only in the first part of the proof of the following proposition. 
As we will see, including this condition does not change the arguments, hence we will omit it in the further,
although the decomposition results also hold for g-hyperbent
functions.
\begin{proposition}
\label{thm-tgen}
Let $k\ge 2t$, and let $f:\F_{2^n}\rightarrow\Z_{2^k}$
be a g-hyperbent function given as
\[ f(x) = a_1(x)+2a_2(x) + \cdots + 2^{k-1}a_k(x) = g(x) + 2^t h(x) \]
for some Boolean functions $a_i:\F_{2^n}\rightarrow\F_2$, $1\le i \le k$, and
\begin{align*}
g(x) &= a_1(x)+2a_2(x) + \cdots + 2^{t-1}a_t(x) \in \mathcal{GB}_n^{2^t}, \\
h(x) &= a_{t+1}(x) + 2a_{t+2}(x) + \cdots + 2^{k-t-1}a_k(x) \in \mathcal{GB}_n^{2^{k-t}}.
\end{align*}
If $n$ is even or $k\ge 3$, then the functions $h(x)$ and $h(x) + 2^{k-2t}g(x)$ are g-hyperbent functions in $\mathcal{GB}_n^{2^{k-t}}$.
\end{proposition}
\begin{proof}
For an integer $i$, $\gcd(i,2^n-1)=1$, and an element $u\in\V_n=\F_{2^n}$, let $f_{u,i}(x) = f(x) + 2^{k-1}\Tr(ux^i)$, $h_{u,i}(x) = h(x) + 2^{k-t-1}\Tr(ux^i)$,
and for $0\le e\le 2^t-1$, $0\le r\le 2^{k-t}-1$, denote by $S^{(u,i)}(e,r)$ the set
\[ S^{(u,i)}(e,r) = \{x\,:\,f_{u,i}(x) = e+2^tr\} = \{x\,:\,g(x) = e,h_{u,i}(x) = r\}. \]
First we suppose that $n$ is even. Then, since $f$ is g-hyperbent, by an obvious version of Lemma~\ref{valdis} for g-hyperbent functions,
for $0\le e\le 2^t-1$ and $0\le \tilde{r}\le 2^{k-t-1}-1$ we have
\[ |S^{(u,i)}(e,\tilde{r})| = |S^{(u,i)}(e,\tilde{r}+2^{k-t-1})| \]
for all but one pair, say the pair $(e,\tilde{r}) = (\epsilon_{u,i},\rho_{u,i})$, for which we have
\[ |S^{(u,i)}(\epsilon_{u,i},\rho_{u,i}+2^{k-t-1})| = |S^{(u,i)}(\epsilon_{u,i},\rho_{u,i})| \pm 2^{n/2}. \]
Consequently,
\begin{align*}
\mathcal{H}_{h,i}^{(2^{k-t})}(u) & = \sum_{x\in\V_n}\zeta_{2^{k-t}}^{h(x)+2^{k-t-1}\Tr(ux^i)} = \sum_{x\in\V_n}\zeta_{2^{k-t}}^{h_{u,i}(x)} \\
& =
\sum_{0\le e\le 2^t-1\atop 0\le r\le 2^{k-t}-1}|S^{(u,i)}(e,r)|\zeta_{2^{k-t}}^r \\
& = \sum_{0\le e\le 2^t-1\atop 0\le \tilde{r}\le 2^{k-t-1}-1}\left[|S^{(u,i)}(e,\tilde{r})|-|S^{(u,i)}(e,\tilde{r}+2^{k-t-1})|\right]\zeta_{2^{k-t}}^{\tilde{r}}\\
&= \pm 2^{n/2}\zeta_{2^{k-t}}^{\rho_{u,i}},
\end{align*}
hence $h$ is g-hyperbent. For $h+2^{k-2t}g$ we have
\begin{align*}
\mathcal{H}_{h+2^{k-2t}g,i}^{(2^{k-t})}(u) & = \sum_{x\in\V_n}\zeta_{2^{k-t}}^{h_{u,i}(x)+2^{k-2t}g(x)} \\
& = \sum_{0\le e\le 2^t-1\atop 0\le \tilde{r}\le 2^{k-t-1}-1}\left[|S^{(u,i)}(e,\tilde{r})|-|S^{(u,i)}(e,\tilde{r}+2^{k-t-1})|\right]\zeta_{2^{k-t}}^{\tilde{r}+2^{k-2t}e} \\
& = \pm 2^{n/2}\zeta_{2^{k-t}}^{\rho_{u,i}+2^{k-2t}\epsilon_{u}},
\end{align*}
and hence $h+2^{k-2t}g$ is g-hyperbent. \\
Now suppose that $n$ is odd and $k\ge 3$. Let $f_{u}(x) = f(x) + 2^{k-1}\Tr(ux)$, $h_{u}(x) = h(x) + 2^{k-t-1}\Tr(ux)$,
$S^{(u)}(e,r) = \{x\,:\,f_{u}(x) = e+2^tr\} = \{x\,:\,g(x) = e,h_{u}(x) = r\}$. If $f$ is gbent, by Proposition \ref{GB-con} there exist two integers
\begin{align*}
\rho_u^{(1)} & = \epsilon_{u,1}+2^t\rho_{u,1} = \rho_u-2^{k-3}, \\
\rho_u^{(2)} & = \epsilon_{u,2}+2^t\rho_{u,2} = \rho_u+2^{k-3},
\end{align*}
where $0\le \epsilon_{u,j}\le 2^t-1$, $0\le \rho_{u,j}\le 2^{k-t-1}-1$, $j=1,2$, such that
\[ |S^{u}(\epsilon_{u,j},\rho_{u,j}+2^{k-t-1})| = |S^{u}(\epsilon_{u,j},\rho_{u,j})|\pm 2^{\frac{n-1}{2}},\, j = 1,2. \]
For $(e,r) \ne (\epsilon_{u,j},\rho_{u,j})$ we have
\[ |S^{u}(e,r+2^{k-t-1})| = |S^{u}(e,r)|. \]
Observe that $\rho_u^{(2)} - \rho_u^{(1)} = \epsilon_{u,2} - \epsilon_{u,1} + 2^t(\rho_{u,2} - \rho_{u,1}) = 2^{k-2}$,
therefore $2^t|(\epsilon_{u,2}-\epsilon_{u,1})$, and consequently $\epsilon_{u,2}=\epsilon_{u,1}$ and $\rho_{u,2}-\rho_{u,1} = 2^{k-t-2}$.
For the generalized Walsh-Hadamard transform of $h$ we then get
\begin{align*}
\mathcal{H}_h^{(2^{k-t})}(u) & = \sum_{x\in\V_n}\zeta_{2^{k-t}}^{h_{u}(x)} =
\sum_{0\le e\le 2^t-1\atop 0\le r\le 2^{k-t}-1}|S^{(u)}(e,r)|\zeta_{2^{k-t}}^r = 2^{\frac{n-1}{2}}\left(\pm\zeta_{2^{k-t}}^{\rho_{u,1}}\pm \zeta_{2^{k-t}}^{\rho_{u,2}}\right) \\
& = 2^{\frac{n-1}{2}}\left(\pm\zeta_{2^{k-t}}^{\rho_{u,1}}\pm \zeta_{2^{k-t}}^{\rho_{u,1}}\zeta_{2^{k-t}}^{2^{k-t-2}}\right)
= 2^{\frac{n-1}{2}}\zeta_{2^{k-t}}^{\rho_{u,1}}(\pm 1\pm i),
\end{align*}
hence $h$ is gbent. For $h+2^{k-2t}g$, using that $\epsilon_{u,2}=\epsilon_{u,1}:=\epsilon_u$ we obtain
\begin{align*}
\mathcal{H}_{h+2^{k-2t}g}^{(2^{k-t})}(u) \
& = \sum_{x\in\V_n}\zeta_{2^{k-t}}^{h_{u}(x)+2^{k-2t}g(x)}\\
&= 2^{\frac{n-1}{2}}\left(\pm\zeta_{2^{k-t}}^{\rho_{u,1}+2^{k-2t}\epsilon_{u}}\pm \zeta_{2^{k-t}}^{\rho_{u,2}+2^{k-2t}\epsilon_{u}}\right) \\
& = 2^{\frac{n-1}{2}}\zeta_{2^{k-t}}^{\rho_{u,1}+2^{k-2t}\epsilon_{u}}(\pm 1\pm i),
\end{align*}
and hence $h+2^{k-2t}g$ is gbent.
\end{proof}
With Proposition \ref{thm-tgen} we can conclude the equivalence of the conditions in Theorem \ref{k,k-1Thm} also for odd $n$.
We use multivariate notation, but keep in mind that many results also apply to g-hyperbent functions, which are only defined
when $\V_n = \F_{2^n}$.
\begin{corollary}
\label{iff-cor}
Let $f\in\mathcal{GB}_n^{2^k}$ with $f(\xx)=g(\xx)+2h(\xx), g\in\cB_n,h\in\cGB_n^{2^{k-1}}$.
Let $n$ be even or $k\ge 3$, then the following statements are equivalent.
\begin{itemize}
\item[$(i)$] $f$ is gbent in $\cGB_n^{2^{k}}$;
\item[$(ii)$] $h$ and $h+2^{k-2}g$ are both gbent in $\cGB_n^{2^{k-1}}$ with $\cH_{h+2^{k-2}g}(\uu)=\pm \cH_{h}(\uu)$
for all $\uu\in \V_n$.
\end{itemize}
\end{corollary}
\begin{proof}
For even $n$, the corollary is Theorem~\ref{k,k-1Thm}. By Remark~\ref{r4odd}, for odd $n$ it suffices to show that
$h$ and $h+2^{k-2}g$ are both gbent in $\cGB_n^{2^{k-1}}$ if $f$ is gbent in $\cGB_n^{2^{k}}$. This follows for
$k\ge 3$ from Proposition \ref{thm-tgen} with $t=1$.
\end{proof}
We can now show one of our main theorems about the decomposition of g-(hyper)bent functions. Proposition~\ref{gbebe}$(i)$, that is, Theorem 18 in~\cite{mms0}, will also follow from this theorem as a special case.
\begin{theorem}
Let $f\in\mathcal{GB}_n^{2^k}$, $k\ge 2$, with $f(\xx)=a_1(\xx)+2a_2(\xx) + \cdots + 2^{k-1}a_k(\xx)$,
$a_i\in\cB_n$, $1\leq i\leq k$, be a gbent function. Let $1\leq s\leq k$, and let $\cc \in \F_2^{s-1}$. The function
\[ g_\cc(x) = a_s+2a_{s+1}+\cdots+2^{k-s} \left(\bigoplus_{i=1}^{s-1} c_i a_i\+ a_k\right) \]
is a gbent function in $\cGB_n^{2^{k-s+1}}$ if
\begin{itemize}
\item[-] $n$ is even,
\item[-] $n$ is odd and $s<k$.
\end{itemize}
Moreover, for $\cc_0 = (c_1,\ldots,c_{s-2},0)$, $\cc_1 = (c_1,\ldots,c_{s-2},1)$ we have
\[ \cH_{g_{\cc_1}}(\uu)=\pm\cH_{g_{\cc_0}}(\uu), \]
for all $\uu\in\V_n$. \\
If $n$ is odd and $s=k$ (hence $g_\cc$ is Boolean), then $g_\cc$ is semibent.
\end{theorem}
\begin{proof}
We show the result by induction. If $s=1$, the claim is obvious. If $s=2$, by taking $g:=a_1$, $h:=a_2+2a_3+\cdots+2^{k-2}a_k$, the claim
follows from Corollary \ref{iff-cor}, since then $f$ is gbent if and only if both $h=a_2+2a_3+\cdots+2^{k-2}a_k$, $h+2^{k-2} g=a_2+2a_3+\cdots+2^{k-2}(a_1\+a_k)$
are gbent and $\cH_h(\uu)=\pm \cH_{h+2^{k-2} g}(\uu)$, for all $\uu$.
Assume the result is true for some $s\le k-1$, i.e., $g_\cc(\xx) = a_s+2a_{s+1}+\cdots+2^{k-s} \left(\bigoplus_{i=1}^{s-1} c_i a_i\+ a_k\right)$
is a gbent function in $\mathcal{GB}_n^{2^{k-s+1}}$ for all $\cc = (c_1,\ldots,c_{s-1})\in\F_2^{s-1}$. We show that it then also holds for $s+1$.
We apply Corollary \ref{iff-cor} to $g_\cc\in\mathcal{GB}_n^{2^{k-s+1}}$. Note that we therefore require $k-s+1 \ge 3$, i.e., $s\le k-2$, if $n$ is odd.
We obtain that for $(\cc,0)$ and $(\cc,1)$ in $\F_2^{s-1}$, both
\[ g_{(\cc,0)} = a_{s+1}+2a_{s+2}+\cdots+2^{k-s-1} \left(\bigoplus_{i=1}^{s-1} c_i a_i\+ a_k\right) \]
and
\[ g_{(\cc,1)} = a_{s+1}+2a_{s+2}+\cdots+2^{k-s-1} \left(\bigoplus_{i=1}^{s-1} c_i a_i\+a_s\+ a_k\right) \]
are gbent functions in $\mathcal{GB}_n^{2^{k-s}}$. Therefore $g_\cc\in\mathcal{GB}_n^{2^{k-s}}$ is gbent for every $\cc\in\F_2^s$,
$1\le s\le k$ when $n$ is even, and $1\le s\le k-1$ when $n$ is odd.
Moreover, again applying Corollary~\ref{iff-cor}, we get
\[ 
\cH_{g_{(\cc,1)}}(\uu)=\pm\cH_{g_{(\cc,0)}}(\uu),
 \]
for all $\uu\in\V_n$. 
By Theorem \ref{sbent}, if $n$ is odd and $s=k$ then for every $\cc\in\F_2^{s-1}$ the Boolean function $g_\cc$ is semibent.
\end{proof}

\begin{remark}
If, conversely, $\cH_{g_{(\cc,1)}}(\uu)=\pm\cH_{g_{(\cc,0)}}(\uu)$ holds for all $\uu\in\V_n$ and $\cc\in\F_2^{s-1}$,
by Corollary~\textup{\ref{iff-cor}}, all functions $g_\cc(\xx) = a_s(\xx)+2a_{s+1}(\xx)+\cdots+2^{k-s} \left(\bigoplus_{i=1}^{s-1} c_i a_i(\xx)\+ a_k(\xx)\right) \in \mathcal{GB}_n^{2^{k-s+1}}$
are gbent. However, we have to impose the analog property for this set of gbent functions for the next step.
\end{remark}

We finish with a decomposition of gbent functions in
$\mathcal{GB}_n^{2^{lt}}$ into gbent functions in
$\mathcal{GB}_n^{2^t}$. The following theorem generalizes both,
Proposition~\ref{gbebe}$(i)$, \cite[Theorem 12]{mms0} and partially
Theorem \ref{ghyperbent}. To this end, let us introduce additional
notation and present some facts that shall help us in our analysis. The
core of the proof of Theorem \ref{ghyperbent} is (\ref{eq:extension})
which simply expresses the following decomposition of
$\zetak^{2^j\alpha}$ for $\alpha\in\{0,1\}$ with respect to
$\{1,\zetak[2]=-1\}$. In fact, equation (\ref{eq:extension}) is simply
a particular case. Indeed, one can express more generally
$\zetak^{2^j\alpha}$, $\alpha\in\ZZ{2^t}$, with respect to
$\{1,\ldots,\zetak[2^t]\}$ if $t$ is a divisor of $k$.  Let
$\mathcal V_{2^t}(\zetak[2^t])$ and
$\mathcal V_{2^t}(\zetak[2^t]^{-1})$ be the $2^t\times 2^t$
Vandermonde matrices :
\begin{align*}
  \mathcal V_{2^t}(\zetak[2^t])=\left(\begin{array}{cccc}
                     1 & 1 & \cdots & 1 \\
                     1 & \zetak[2^t] &  & \zetak[2^t]^{2^t-1} \\
                     \vdots & \vdots & & \vdots \\
                     1 & \zetak[2^t]^{2^t-1} & & \zetak[2^t]^{(2^t-1)(2^t-1)} \\
                   \end{array}
  \right)
\end{align*}
and
\begin{align*}
  \mathcal V_{2^t}(\zetak[2^t]^{-1})=\left(\begin{array}{cccc}
                     1 & 1 & \cdots & 1 \\
                     1 & \zetak[2^t]^{-1} &  & \zetak[2^t]^{-(2^t-1)} \\
                     \vdots & \vdots & & \vdots \\
                     1 & \zetak[2^t]^{-(2^t-1)} & & \zetak[2^t]^{-(2^t-1)(2^t-1)}
                   \end{array}
  \right).
\end{align*}
Observe that
\begin{align}\label{eq:invandermonde}
  \mathcal V_{2^t}(\zetak[2^t]) \mathcal V_{2^t}(\zetak[2^t]^{-1}) = 2^t\mathbf{I}_{2^t},
\end{align}
where $\mathbf{I}_{2^t}$ stands for the identity matrix of size $2^t$.
Define now a collection of maps from $\mathbb C$ to itself by setting
\begin{align*} 
  \left(
  \begin{array}{c}
    h_0(z) \\ h_1(z) \\ \vdots \\ h_{2^t-1}(z)
  \end{array}
  \right)
  =
  \mathcal V_{2^t}(\zetak[2^t]^{-1})
  \left(
  \begin{array}{c}
    1 \\ z \\ \vdots \\ z^{2^t-1}
  \end{array}
  \right)
\end{align*}
or equivalently, for any $\alpha\in\ZZ{2^t}$,
\begin{align}\label{eq:expressionhz}
  h_\alpha(z) 
  &= \sum_{\beta\in\ZZ{2^t}} \zetak[2^t]^{-\alpha\beta}z^\beta.
\end{align}
Furthermore, according to (\ref{eq:invandermonde}), one has,
for any $z\in\C$, 
\begin{align}\label{eq:decompzetakt}
   \left(
  \begin{array}{c}
    1 \\ z \\ \vdots \\ z^{2^t-1}
  \end{array}
  \right)
  =
  \frac{1}{2^t}\mathcal V_{2^t}\left(\zetak[2^t]\right)
  \left(
  \begin{array}{c}
    h_0(z) \\ h_1(z) \\ \vdots \\ h_{2^t-1}(z)
  \end{array}
  \right)
\end{align}
that is, for $\beta\in\ZZ{2^t}$,
\begin{align}
  \label{eq:expressionzh}
  z^\beta = \frac{1}{2^t} \sum_{\alpha\in\ZZ{2^t}}\zetak[2^t]^{\alpha\beta}h_{\beta}(z).
\end{align}
Then, we show the next theorem.
\begin{theorem}
  Let $n$ be even. Let $k=lt$ and let $f\in\mathcal{GB}_n^{2^k}$ be a
  g-hyperbent function given as
\[ f(\xx) = b_1(\xx) + 2^tb_2(\xx) + \cdots + 2^{(l-1)t}b_l(\xx), \]
for some functions $b_i\in\mathcal{GB}_n^{2^t}$, $1\le i\le l$. If $n$ is even or $t\ge 2$, then for every
$\cc = (c_1,c_2,\ldots,c_{l-1}) \in \Z_{2^t}^{l-1}$, the function
\[ g_{\cc}(\xx) = c_1b_1(\xx) + \cdots + c_{l-1}b_{l-1}(\xx) + b_l(\xx) \in \mathcal{GB}_n^{2^t} \]
is g-hyperbent.
\end{theorem}
\begin{proof}
  The extended Hadamard-Walsh transform of $f$ is :
  \begin{align}\label{eq:WHt}
  \GWa[,i]{2^k}{f}(a) 
    &= \sum_{x\in\GF n}\zetak^{f(x)}(-1)^{\Tr(ax^i)}.
  \end{align}
Using~(\ref{eq:expressionhz}) and~(\ref{eq:expressionzh}), one gets
\begin{align*}
  \zetak^{f(x)} &= \zetak[2^t]^{b_l(x)}\prod_{j=1}^{l-1}\zetak^{2^{(j-1)t}b_j(x)}
  =   \zetak[2^t]^{b_l(x)} \prod_{j=1}^{l-1}\Bigg( \frac{1}{2^t}\sum_{c_j\in\ZZ{2^t}}\zetak[2^t]^{c_j b_j(x)}h_{c_j}\left(\zetak^{2^{(j-1)t}}\right)\Bigg)
    \\
   &=  \zetak[2^t]^{b_l(x)}\prod_{j=1}^{l-1}\Bigg( \frac{1}{2^t}\sum_{c_j\in\ZZ{2^t}}\zetak[2^t]^{c_j b_j(x)} \sum_{d_j\in\ZZ{2^t}} \zetak[2^t]^{-c_jd_j}\zetak^{2^{(j-1)t}d_j}
  \Bigg)\\
  &=   \zetak[2^t]^{b_l(x)}\prod_{j=1}^{l-1}\Bigg(\frac{1}{2^t}\sum_{(c_j,d_j)\in\ZZ{2^t}^2}\zetak[2^t]^{c_j (b_j(x)-d_j)}\zetak^{2^{(j-1)t}d_j}\Bigg).
\end{align*}
Therefore,
\begin{align*}
  \zetak^{f(x)} &= \zetak[2^t]^{b_l(x)}\times 
  \frac1{2^{k-t}}\sum_{(c,d)\in\ZZ{2^t}^{l-1}\times\ZZ{2^t}^{l-1}} \zetak^{\sum_{j=1}^{l-1}2^{(j-1)t}d_j}  \zetak[2^t]^{\sum_{j=0}^{l-1}c_j(b_j(x)-d_j)}\\
  &= \frac1{2^{k-t}}\sum_{(c,d)\in\ZZ{2^t}^{l-1}\times\ZZ{2^t}^{l-1}} \zetak^{\sum_{j=1}^{l-1}2^{(j-1)t}d_j} \zetak[2^t]^{g_{\cc}(x)-c\cdot d},
\end{align*}
where $c\cdot d=\sum_{j=1}^{l-1}c_jd_j$. If we use the above relation in ~(\ref{eq:WHt}), then
\begin{align}\label{eq:relationGWa}
  \GWa[,i]{2^k}{f}(a) 
  &= \frac1{2^{k-t}}\sum_{(c,d)\in\ZZ{2^t}^{l-1}\times\ZZ{2^t}^{l-1}}\zetak[2^t]^{-c\cdot d}\,\zetak^{\sum_{j=1}^{l-1}2^{(j-1)t}d_j}\,\Wa[,i]{f_c}(a).
\end{align}
At this stage, observe that~(\ref{eq:relationGWa})
generalizes~(\ref{eq:relationGWaGray}) (which corresponds to
$t=1$). It is then quite straightforward to repeat the arguments of the
proof of Theorem~\ref{ghyperbent} to get Theorem~\ref{thm-tgen}.
\end{proof}

\section{Conclusion}

In this paper we extend the concept of a hyperbent function to generalized Boolean functions from $\F_{2^n}$ to $\Z_{2^k}$, and we present examples 
of generalized hyperbent functions obtained with partial spreads. We investigate decompositions of generalized (hyper)bent functions (gbent respectively g-hyperbent functions).
We prove that g-(hyper)bent functions from $\F_{2^n}$ to $\Z_{2^k}$ decompose into g-(hyper)bent functions from $\F_{2^n}$ to $\Z_{2^{k^\prime}}$ 
for some $k^\prime < k$. We show that when $n$ is odd, then the Boolean functions associated to a generalized bent function form an affine space of semibent functions. 
This complements a result \cite{mms0}, where it is shown that for even $n$ the associated Boolean functions are bent.

We finally remark that for a gbent function from $\V_n$ to $\Z_{2^k}$, the function $cf$ is in general not gbent when $c\in\Z_{2^k}$
is even. Functions for which $cf$ is gbent for every nonzero $c$ seem to be quite rare. Some examples obtained from partial 
spreads are in \cite{mms}. Such functions may be particularly interesting for future research as they yield relative difference sets
(being bent). For a general discussion on relative difference sets and functions between arbitrary abelian groups we refer to \cite{p04}. \\[.5em]

\noindent
{\bf Acknowledgement.} The second author is supported by the Austrian Science Fund (FWF) Project no. M 1767-N26.

\end{document}